\documentclass[12pt]{article}

\usepackage{amsmath,amsthm,amssymb,amscd}
\usepackage{graphicx, float, epsfig}
\usepackage{rotating}
\usepackage[hidelinks]{hyperref}
\usepackage[semicolon,comma]{natbib}
\usepackage{ds}

\usepackage{authblk}

\usepackage{enumitem} % using indent items in enumerate
\usepackage{caption, subcaption} % using the caption for tables and figures
\usepackage{capt-of} % cross-referencing the figures in appendix
\usepackage[flushleft]{threeparttable} % adding the footnote in a table
\usepackage{multirow}
\usepackage{booktabs}

\usepackage{xcolor,soul} % adding color fonts and highlight text

\newtheorem{Lemma}{Lemma} %numbering by section
\newtheorem{theorem}{Theorem} %numbering by default: (1), (2),...
\usepackage{centernot}

\author[1]{Sheng Dai}
    \author[2]{Timo Kuosmanen}
    \author[3]{Xun Zhou}
    \affil[1~]{School of Economics, Zhongnan University of Economics and Law, 430073 Wuhan, China}
    \affil[2~]{Turku School of Economics, University of Turku, 20500 Turku, Finland}
    \affil[3~]{Surrey Business School, University of Surrey, Guildford GU2 7XH, UK}
\title{\bf Orthogonality conditions for convex regression\footnote{
    \textit{E-mail addresses:} \texttt{sheng.dai@zuel.edu.cn (S. Dai)}, \texttt{timo.kuosmanen@utu.fi (T. Kuosmanen)},\\
    \hspace*{34mm} \texttt{x.zhou@surrey.ac.uk (X. Zhou)}.}}
\date{May 30, 2025}

\begin{document}
% revising the caption Figure XX: to Fig. XX.; Table XX: to Table XX.
\captionsetup[figure]{labelfont={bf},labelformat={default},labelsep=period,name={Fig.}}
\captionsetup[table]{labelfont={bf},labelformat={default},labelsep=period,name={Table}}

\maketitle

\vfill

\vfill

\begin{abstract}
\noindent Econometric identification generally relies on orthogonality conditions, which usually state that the random error term is uncorrelated with the explanatory variables. In convex regression, the orthogonality conditions for identification are unknown. Applying Lagrangian duality theory, we establish the sample orthogonality conditions for convex regression, including additive and multiplicative formulations of the regression model, with and without monotonicity and homogeneity constraints. We then propose a hybrid instrumental variable control function approach to mitigate the impact of potential endogeneity in convex regression. The superiority of the proposed approach is shown in a Monte Carlo study and examined in an empirical application to Chilean manufacturing data.
\\[5mm]
\end{abstract}
\vfill

\thispagestyle{empty}
%-----------------%
%
%-----------------%
\newpage
\setcounter{page}{1}
\setcounter{footnote}{0}
\pagenumbering{arabic}
\baselineskip 20pt

%-----------------%
%
%-----------------%

\section{Introduction}\label{sec:intro}

Shape-constrained nonparametric regression avoids restrictive assumptions about the functional form of the regression function by building upon convexity \citep{Hildreth1954} and/or monotonicity constraints \citep{Brunk1955, brunk1958}. Convexity and monotonicity constraints are particularly relevant in the microeconomic applications where the theory implies certain monotonicity and convexity/concavity properties for many functions of interest \citep[e.g.,][]{Afriat1967, Afriat1972, Varian1982, Varian1984}. For example, the cost function of a firm must be monotonic increasing and convex with respect to the input prices. The recent developments in the convex regression enable researchers to impose concavity or convexity constraints implied by the theory to estimate the functions of interest without any parametric functional form assumptions. Since the development of an explicit piecewise linear characterization by \citet{Kuosmanen2008}, convex regression has attracted growing interest in econometrics, statistics, operations research, machine learning, and related fields \citep[e.g.,][]{Seijo2011, Lim2012, Hannah2013, Mazumder2019, Yagi2018, Bertsimas2020}. Extensions of the convex regression to the quantile and expectile estimation have proved useful for estimating marginal abatement costs of pollution emissions \citep{Kuosmanen2020b, Dai2023, Dai2025}.
 
In the linear regression, the random error term is assumed to be uncorrelated with the explanatory variables. In the generalized method of moments (GMM) interpretation of the linear regression, this property is called the population orthogonality condition. It is easy to show that the residuals of the ordinary least squares (OLS) estimator satisfy the corresponding sample orthogonality condition, that is, OLS residuals do not correlate with any of the regressors by construction. Essential statistical properties such as unbiasedness and consistency of the OLS estimator critically depend on the orthogonality condition. If the orthogonality condition does not hold, the problem is referred to as endogeneity. The standard econometric approach is then to find instrumental variables (IVs) that are highly correlated with the regressors but uncorrelated with the error term.

The orthogonality conditions for convex regression remain unknown. While the statistical consistency of the convex regression has been proved assuming the standard orthogonality condition, we do not know how convex regression behaves if regressors are endogenous. Consequently, practitioners are unable to evaluate the suitability of convex regression in applications where explanatory variables may be endogenous, and they lack the means to address potential endogeneity bias. The contributions of this paper are three-fold:

The first contribution of this paper is to formally analyze the sample orthogonality for convex regression. Applying Lagrange duality, we establish the sample orthogonality conditions for the most essential variants of convex regression, including additive and multiplicative formulations of the regression model, with and without monotonicity and homogeneity constraints. The results of this paper are relevant to practitioners who need to understand what kind of exogeneity assumptions are required for different variants of convex regression. 

The second contribution of this paper is to propose a hybrid IV control function approach integrated into convex regression. The classic two-stage least squares (2SLS) approach can be adapted to correct the endogeneity bias in convex regression, but the nonparametric IV estimation may lead to poor statistical performance due to its ill-posedness \citep{Chetverikov2017}. Moreover, it is difficult to estimate convex regression with endogeneity (a pure nonparametric regression model) because the 2SLS approach is not directly transferable to a nonlinear or nonparametric setting \citep{Yatchew2003}. Alternatively, the control function is commonly used to address the simultaneity bias \citep[see, e.g.,][]{Olley1996, Levinsohn2003, Ackerberg2015} or the measurement error in inputs \citep[see, e.g.,][]{Collard2016, Dong2021} when estimating the production function. In the spirit of fully utilizing the IV and control function, we propose a hybrid IV control function-based convex regression to address the potential endogeneity problem. The consistency of the proposed approach is established.

Our third contribution is to conduct three Monte Carlo experiments in the context of the classic Cobb–Douglas production function. The first experiment examines the finite sample performance of OLS and convex regression estimators under endogeneity. The second experiment compares the finite sample performance of convex regression and convex regression with the 2SLS and IV control function approaches, where the monotonicity is relaxed. The third experiment tests whether the performance of convex regression with monotonicity improves or deteriorates when the IV and IV control function approaches are introduced. An empirical application to the Chilean manufacturing production data illustrates whether the hybrid IV control function approach empirically makes a difference and whether monotonicity should be avoided in convex regression.

We further shed new light on the question: if the sample orthogonality does not hold, then how does the performance of convex regression with IV control function? In the existing literature, there is little evidence of whether the performance should improve or deteriorate after using the IV control function or even the 2SLS approach. Our simulation findings suggest that the convex regression with monotonicity constraints (i.e., the sample orthogonality condition cannot be held) can improve its performance after using the IV control function approach in the case of endogeneity. 

The rest of this paper is organized as follows. Section~\ref{sec:cnls} briefly introduces the additive and multiplicative convex regression. Section~\ref{sec:condit} examines the sample orthogonality conditions in the cases with and without monotonicity. The two-stage IV control function-based convex regression is developed in Section~\ref{sec:cf}. Sections~\ref{sec:mc} and \ref{sec:appl} implement the Monte Carlo simulations and empirical application to Chilean manufacturing data. Section~\ref{sec:conc} concludes the paper with suggested future avenues. The formal proofs and additional tables and figures are attached in the Appendix. 

%-----------------%
%
%-----------------%

\section{Preliminaries on convex regression}\label{sec:cnls}

\subsection{Additive model}

Consider the following multivariate additive model
\begin{equation}
	y_i = f(\bx_i) + \varepsilon_i
	\label{eq:eq1}
\end{equation}
where $y_i$ is the dependent variable, $\bx_i$ is a $k$-dimensional vector of explanatory variables, and $\varepsilon_i$ is a random error term with zero mean and a constant and finite variance. $f:\real_+^k \rightarrow \real$ is assumed to be a (homogeneous) concave/convex regression function with an unknown functional form \citep[see, e.g.,][]{Kuosmanen2008, Seijo2011, Kuosmanen2012c}. 

Building on the results by \citeauthor{Afriat1967} (\citeyear{Afriat1967}, \citeyear{Afriat1972}), \citet{Kuosmanen2008} develops the first operational solution to the multivariate convex regression problem.\footnote{
    In this paper, we focus on a concave $f$, but a convex $f$ can be modeled by reversing the sign of the Afriat inequalities \citep{Afriat1967, Afriat1972}. We use the term ``convex regression'' because in both cases $f$ is a support of a convex set, whereas concave sets do not exist. See \citet{Liao2024} for the case of convex $f$.
}
Specifically, we consider the least squares estimator of multivariate convex regression

\begin{alignat}{2}
    \min \, &\dfrac{1}{2}\sum_{i=1}^n\varepsilon_i^2 \label{eq:eq2} \\
    \mbox{\textit{s.t.}}\quad 
    &  y_i = f(\bx_i) + \varepsilon_i\text{  }\forall i  \tag{2a} \label{eq:2a} \\
    &  f \text{ } \text{is concave} \tag{2b} \label{eq:2b} \\
    &  f \text{ } \text{is monotonic increasing} \tag{2c} \label{eq:2c} \\
    &  f \text{ } \text{is linearly homogeneous} \tag{2d} \label{eq:2d}
\end{alignat}
where concavity \eqref{eq:2b} is imposed throughout the paper, but monotonicity \eqref{eq:2c} and homogeneity \eqref{eq:2d} are optional. 

\citet{Kuosmanen2008} also proves that the optimal solution to the infinite-dimensional optimization problem \eqref{eq:eq2} is always equivalent to the optimal solution to the following quadratic programming problem with linear constraints (also referred to as the convex nonparametric least squares estimator)
\begin{alignat}{2}
    \underset{\alpha, \bbeta, \varepsilon}{\min} \, & \dfrac{1}{2}\sum_{i=1}^n\varepsilon_i^2 & \quad & \tag{3} \label{eq:eq3} \\
    \text{\textit{s.t.}}\quad 
    & y_i = \alpha_i + \bbeta_i^\prime\bx_i + \varepsilon_i &{\quad}& \forall i  \tag{3a} \label{eq:3a} \\
    & \alpha_i + \bbeta_i^\prime\bx_i \le \alpha_h + \bbeta_h^\prime\bx_i &{\quad}&  \forall i, h   \tag{3b} \label{eq:3b} \\
    & \bbeta_i \ge \bzero &{\quad}& \forall i \tag{3c} \label{eq:3c}\\
    & \alpha_i = 0 &{\quad}& \forall i \tag{3d} \label{eq:3d}
\end{alignat}
where the first set of constraints \eqref{eq:3a} restates the regression equation \eqref{eq:eq1} in terms of a piecewise linear approximation of the true but unknown regression function $f$. \citet{Kuosmanen2008} has proved that without any loss of generality, the estimated regression function can be parametrized as a piecewise linear function and characterized by the coefficients $\alpha$ and $\bbeta$.\footnote{
    It is worth noting that there exist alternative parametrizations of concavity (e.g., used by \citealp{Seijo2011}); the advantage of the applied parametrization in this paper is that we can easily impose monotonicity and linear homogeneity as well. 
}
The second set of constraints \eqref{eq:3b} enforces the concavity of the piecewise linear regression function (reversing the sign of the inequality imposes convexity). The third set of constraints \eqref{eq:3c} enables the regression function to be monotonic increasing (reversing the sign makes the function monotonic decreasing). The fourth constraint set in \eqref{eq:3d} enforces linear homogeneity of the regression function.

\subsection{Multiplicative model}
\setcounter{equation}{3} 

While the literature on convex regression focuses almost exclusively on the additive model introduced above, in economic applications, it is often more natural to posit a multiplicative model (\citeauthor{Kuosmanen2012c}, \citeyear{Kuosmanen2012c})
\begin{equation}
    y_i = f(\bx_i)\exp({\varepsilon_i})
    \label{eq:eq2.2.1}
\end{equation}

Taking the natural logarithm of both sides yields
\begin{equation}
    \ln{y_i} = \ln{f(\bx_i)} + \varepsilon_i,
    \label{eq:eq2.2.2}
\end{equation}
where the multiplicative formulation is not just an inconsequential data transformation, because $\ln{f(\bx_i)} \neq f(\ln({\bx_i}))$ and hence the orthogonality conditions of the multiplicative formulation deserve a separate treatment.

Accordingly, the multiplicative convex regression estimator is formulated as 
\begin{alignat}{2}
    \underset{\alpha, \bbeta, \varepsilon}\min \, &\dfrac{1}{2}\sum_{i=1}^n\varepsilon_i^2 &{\quad}& \label{eq:eq4}\\
    \mbox{\textit{s.t.}}\quad 
    &  \ln y_i=\ln (\alpha_i+\bbeta^\prime_i\bx_i)+\varepsilon _{i}^{{}} &{\quad}& \forall i    \tag{6a} \label{eq:5a} \\
    &  \alpha_i + \bbeta_i^\prime\bx_i \le \alpha_h + \bbeta_h^\prime\bx_i  &{\quad}&  \forall i, h   \tag{6b} \label{eq:5b} \\
    &  \bbeta_i \ge \bzero &{\quad}& \forall i  \tag{6c} \label{eq:5c} \\
    & \alpha_i = 0 &{\quad}& \forall i \tag{6d} \label{eq:5d}
\end{alignat}
where the first set of constraints \eqref{eq:5a} restates the logarithm of the multiplicative function model \eqref{eq:eq2.2.2}, and the other three constraints \eqref{eq:5b}-\eqref{eq:5d} are also used to guarantee the concavity, monotonicity, and linear homogeneity of the regression function $f$, respectively. 

%-----------------%
%
%-----------------%

\section{Orthogonality conditions}\label{sec:condit}

In this section, we systematically analyze the sample orthogonality conditions of convex regression in various cases, including the additive and multiplicative models with and without monotonicity.

\subsection{Convex regression subject to concavity}\label{orth1}

We first consider the additive convex regression problem~\eqref{eq:eq3} without monotonic and linearly homogeneous assumptions (i.e., excluding constraints \eqref{eq:3c} and \eqref{eq:3d}). The Lagrangian function of convex regression without monotonicity and linear homogeneity is stated as 
\begin{align*}
    \Ls(\blambda,\balpha,\bbeta,\bmu)=&\dfrac{1}{2}\sum\limits_{i=1}^n \varepsilon_i^2+\sum\limits_{i=1}^n\mu_i(y_i-\alpha_i-\bbeta_i^\prime\bx_i-\varepsilon_i) \\
    & + \sum\limits_{i=1}^n{\sum\limits_{h=1}^n\lambda_{ih}\left[\alpha_i-\alpha_h+(\bbeta _i-\bbeta_h)^\prime\bx_i \right]}
\end{align*}
where $\lambda_{ih} \ge 0$ and $\mu_i$ are the Lagrangian multipliers (shadow prices) of the concavity constraint~\eqref{eq:3b} and regression equation~\eqref{eq:3a}. Problem \eqref{eq:eq3} is convex, and the optimal solution exists and is unique. Furthermore, the optimal solution of problem \eqref{eq:eq3} satisfies the first-order conditions, which are the first-order derivatives of the Lagrangian function $\Ls$,
\begin{align}
& \dfrac{\partial \Ls}{\partial \varepsilon_i}=e_i - \mu_i=0, \forall i \label{eq:eq5}\\
& \dfrac{\partial \Ls}{\partial \alpha_i}=-\mu_i+\sum\limits_{h=1}^n \lambda_{ih} -\sum\limits_{h=1}^n \lambda_{hi}=0, \forall i \label{eq:eq6}\\
& \dfrac{\partial \Ls}{\partial \beta_{ik}}=-\mu_i x_{ik}+\sum\limits_{h=1}^n \lambda_{ih} x_{ik} - \sum\limits_{h=1}^n\lambda_{hi} x_{hk}=0, \forall i, k. \label{eq:eq7}
\end{align}
where $e_i$ is the estimates of $\varepsilon_i$.\footnote{
    In order to keep the symbols easy to read, we avoid using other symbols for the estimated values of other variables such as $\alpha$ and $\bbeta$.
}
Furthermore, the optimal solution necessarily satisfies the complementary slackness conditions (i.e., the Karush-Kuhn-Tucker conditions), 
\begin{equation}
    \label{eq:eq8}
    \lambda_{ih}[\alpha_i-\alpha_h+(\bbeta_i-\bbeta_h)^\prime\bx_i]=0
\end{equation}

Conditions \eqref{eq:eq5}-\eqref{eq:eq8} make a foundation to prove several properties of convex regression.
\begin{Lemma}
    In the optimal solution to the convex regression problem~\eqref{eq:eq3}, the sum of residuals is equal to zero, $\sum\limits_{i=1}^ne_i=0$.
    \label{prop1}
\end{Lemma}
\begin{proof}
    See Appendix~\ref{sec:proofs}.
\end{proof}

\citet{Seijo2011} apply Moreau's decomposition theorem to prove an equivalent result $\sum\limits_{i=1}^n y_i=\sum\limits_{i=1}^n \hat{y}_i$. The proof of Lemma~\ref{prop1} is an intuitive direct proof based on the first-order conditions of convex regression. This strategy also enables us to prove the following Theorem~\ref{theo1}.
\begin{theorem}
    The optimal solution to the additive convex regression problem subject to concavity constraints satisfies the sample orthogonality condition  $\sum\limits_{i=1}^ne_ix_{ik}=0 \,\ \forall k$.
    \label{theo1}
\end{theorem}
\begin{proof}
    See Appendix~\ref{sec:proofs}. 
\end{proof}

This result is important for understanding the underlying statistical assumptions of the convex regression estimator. Recall that the OLS estimator of the linear regression model satisfies exactly the same sample orthogonality condition. The corresponding population orthogonality condition (also referred to as the exogeneity condition) is 
\begin{equation}
    \E\left[\varepsilon_i x_{ik} \right] = 0 \,\ \forall k.
    \label{eq:eq12}
\end{equation}

In other words, the random error term $\varepsilon$ is assumed to be uncorrelated with regressors $\bx$. Exactly the same exogeneity condition applies to convex regression. Hence, the convex regression estimator is subject to endogeneity problems similar to those of the OLS estimator. To address such endogeneity problems, note that the convex regression estimator can be understood as a GMM estimator when the population orthogonality condition \eqref{eq:eq12} holds. The convex regression estimator is also a maximum likelihood estimator if $\varepsilon_i$ is normally distributed, as already noted by \citet{Hildreth1954}. Therefore, it is possible to apply an IV in the context of convex regression, stating the population orthogonality condition \eqref{eq:eq12} in terms of the instrument $z$ that correlates with $x$ but is uncorrelated with $\varepsilon$ (we revisit this in detail below). 

We know from Seijo and Sen (\citeyear{Seijo2011}, Lemma 2.4, part ii) the following weaker result,
\begin{equation*}
    \sum\limits_{i=1}^n \hat{y}_i(y_i-\hat{y}_i)=\sum\limits_{i=1}^n e_i \hat{y}_i=0.
\end{equation*}

Theorem~\ref{theo1} implies this result, but the converse is not true. Note that the population orthogonality corresponding to Seijo and Sen's result is
\begin{equation*}
    \E\left[\varepsilon_i f(\bx_i) \right]=0,
\end{equation*}
which is the standard exogeneity condition in parametric nonlinear regression. Interestingly, Theorem~\ref{theo1} demonstrates that convex regression also satisfies the stronger sample orthogonality condition of linear regression. It is obviously critically important to know exactly what kind of assumptions must be made before applying the estimator. 

Considering next the multiplicative convex regression problem~\eqref{eq:eq4} without monotonic and linear homogeneous assumptions (i.e., excluding constraints \eqref{eq:5c} and \eqref{eq:5d}), we have the following Lagrangian function 
\begin{alignat*}{2}
    \Ls(\blambda,\balpha,\bbeta,\bmu)=&\dfrac{1}{2}\sum\limits_{i=1}^n\varepsilon_i^{2}+\sum\limits_{i=1}^n\mu_i\left[\ln (y_i)-\ln (\alpha_i+\bbeta_i^\prime\bx_i)-\varepsilon_i\right] \\
    &+\sum\limits_{i=1}^n\sum\limits_{h=1}^n\lambda_{ih}\left[\alpha_i-\alpha_h+(\bbeta_i-\bbeta_h)^\prime\bx_i\right].
\end{alignat*}

Denoting $\hat{y}_i = \alpha_i + \bbeta_i^\prime\bx_i$ and differentiating, we have the first-order condition \eqref{eq:eq5} and the following
\begin{align}
	& \dfrac{\partial \Ls}{\partial \alpha_i}=-\mu_i\dfrac{1}{\hat{y}_i}+\sum\limits_{h=1}^n \lambda_{ih}-\sum\limits_{h=1}^n \lambda_{hi}=0 \,\ \forall i, \label{eq:eq16} \\
	& \dfrac{\partial \Ls}{\partial \beta_{ik}}=-\mu_i\dfrac{x_{ik}}{\hat{y}_i}+\sum\limits_{h=1}^n \lambda_{ih} x_{ik} - \sum\limits_{h=1}^n \lambda_{hi}x_{hk} = 0 \,\ \forall i, k.
	\label{eq:eq17} 
\end{align}

Conditions \eqref{eq:eq5} and \eqref{eq:eq16} imply 
\begin{equation*}
	e_i=\hat{y}_{i}\left[ \sum\limits_{h=1}^n\lambda_{ih}-\sum\limits_{h=1}^n\lambda_{hi} \right].
\end{equation*}	

In this case, it is easy to confirm that we cannot conclude a result similar to Lemma \ref{prop1}. That is, the sum of residuals in the optimal solution to \eqref{eq:eq4} is not necessarily zero. However, a close conclusion is stated as
\begin{Lemma}
    In the optimal solution to the multiplicative convex regression problem~\eqref{eq:eq4}, the scaled regression residuals sum to zero, that is, $\sum\limits_{i=1}^n \dfrac{e_i}{\hat{y}_i}=0$.
    \label{prop2}
\end{Lemma}
\begin{proof}
    See Appendix~\ref{sec:proofs}. 
\end{proof}

Lemma~\ref{prop2} is interesting as it builds a framework to study the properties of a multiplicative error model using the ratio of regressors to the fitted value. Based on Lemma~\ref{prop2}, we have the following Theorem~\ref{theo2} to document the sample orthogonality condition for the multiplicative model. 
\begin{theorem}
    The optimal solution to the multiplicative convex regression problem subject to the concavity constraints satisfies the sample orthogonality condition  $\sum\limits_{i=1}^n{e_i\dfrac{x_{ik}}{\hat{y}_i}}=0 \,\ \forall k$.
    \label{theo2}
\end{theorem}
\begin{proof}
    See Appendix~\ref{sec:proofs}. 
\end{proof}

The corresponding population orthogonality condition is 
\begin{equation*}
    \E\left[\varepsilon_i\dfrac{\bx}{f(\bx)} \right]=\mathbf{0}.
\end{equation*}

Therefore, it is not necessary to assume that $\bx$ are uncorrelated with the random error term $\varepsilon$, but rather, we need to assume that ratios $\bx/f(\bx)$ do not correlate with $\varepsilon$. This result can guide one with the choice of good instruments and help to circumvent certain specific types of endogeneity problems (e.g., the simultaneity problem by \citet{Marschak1944} stressed by \citet{Olley1996} among others). For instance, if low productivity firms need to use more inputs $\bx$, but $f(\bx)$ increases by the same proportion, then this would not cause the potential simultaneity problem.

\subsection{Convex regression subject to concavity and monotonicity}\label{orth2}

Suppose the regression function $f$ is concave and monotonic increasing (i.e., excluding constraints \eqref{eq:3d} and \eqref{eq:5d} from additive model \eqref{eq:eq3} and multiplicative model \eqref{eq:eq4}, respectively). Convex regression implements monotonicity by imposing a sign constraint for the slope coefficients $\bbeta_i$. 

In the case of the additive model subject to concavity and monotonicity, the Lagrangian function $\Ls$ becomes 
\begin{alignat*}{2}
    \Ls(\blambda,\balpha,\bbeta,\bmu,{\boldeta})=\dfrac{1}{2}&\sum\limits_{i=1}^n\varepsilon _i^2+\sum\limits_{i=1}^n\mu_i(y_i-\alpha_i-\bbeta_i^\prime\bx_i-\varepsilon_i)  \\
    &+\sum\limits_{i=1}^n{\sum\limits_{h=1}^n{\lambda_{ih}\left[\alpha _i-\alpha _h+(\bbeta_i-\bbeta_h)^\prime\bx_i\right]}}-\sum\limits_{i=1}^n \boldeta_i^\prime\bbeta_i,
\end{alignat*}
where ${\boldeta}_i \ge 0$ are shadow prices of the monotonicity constraints.
 
Differentiating, we obtain the first-order optimality conditions \eqref{eq:eq5}, \eqref{eq:eq6}, and the following
\begin{equation}
    \dfrac{\partial \Ls}{\partial \beta_{ik}}=-\mu_i x_{ik}+\sum\limits_{h=1}^n\lambda_{ih} x_{ik} -\sum\limits_{h=1}^n \lambda_{hi} x_{hk}-\eta_{ik}=0.
    \label{eq:eq13}
\end{equation}

Similarly, the first-order conditions for the multiplicative model subject to concavity and monotonicity are summarized as \eqref{eq:eq5}, \eqref{eq:eq16}, and the following 
\begin{align}
    \dfrac{\partial \Ls}{\partial \beta_{ik}}=-\mu_i\dfrac{x_{ik}}{\hat{y}_i}+\sum\limits_{h=1}^n \lambda_{ih} x_{ik}-\sum\limits_{h=1}^n \lambda_{hi}x_{hk} - \eta_{ik} = 0.
    \label{eq:eq22} 
\end{align}

In this case, conditions \eqref{eq:eq5} and \eqref{eq:eq22} imply that $\sum\limits_{i=1}^n \dfrac{e_i}{\hat{y}_i}=0$, thus Lemma \ref{prop2} is valid. 

In both additive and multiplicative models, if at least for one observation the monotonicity constraint is active (i.e., $\mu_{ij}>0$), then the sample orthogonality condition for convex regression must fail (Theorem~\ref{theo3}). 
\begin{theorem}
    The optimal solution to the convex regression problem subject to concavity and monotonicity satisfies the following inequalities, \\
    \hspace*{2cm} i) Additive model: $\sum\limits_{i=1}^n e_i x_{ik}<0 \,\ \forall k $; \\
    \hspace*{2cm} ii) Multiplicative model: $\sum\limits_{i=1}^n{e_i \dfrac{x_{ik}}{\hat{y}_i}}<0 \,\ \forall k.$
    \label{theo3}
\end{theorem}
\begin{proof}
    See Appendix~\ref{sec:proofs}. 
\end{proof}

The corresponding population orthogonality conditions for convex regression are not necessarily held. Specifically, for the additive model subject to concavity and monotonicity, we have 
$$\E[\varepsilon_i\bx_{ik}]\le0 \,\ \forall k, $$
and for the multiplicative model subject to concavity and monotonicity, there is
$$\E\left[\varepsilon_i\dfrac{\bx}{f(\bx)}\right]\le \mathbf{0}.$$

In fact, the sample orthogonality condition holds if and only if the monotonicity constraint is redundant for all observations. However, it is unlikely in practice that the monotonicity constraints do not bind for any single observation. Even if the true function is monotonic, there are usually some violations due to the random error $\varepsilon$. When monotonicity constraints are binding, the residuals negatively correlate with the explanatory variables. This might suggest that convex regression with monotonicity constraints can accommodate specific types of endogeneity. However, the sum of shadow prices $\sum\limits_{i=1}^n\eta_{ik}$ does not seem to have any obvious econometric meaning, and it is difficult to suggest the exact population orthogonality condition. %We thus recommend practitioners that the monotonicity constraints should not be imposed whenever it is important to ensure that the orthogonality conditions hold, for example, when using instruments. 

If the sample orthogonality does not hold, then what does it imply for the assumptions regarding the true $\varepsilon$? Do we implicitly assume that $\text{Cov}(x, \varepsilon) < 0$? Is that enough? Further, what happens if the usual population orthogonality $\text{Cov}(x, \varepsilon) = 0$ does hold, but one imposes monotonicity? Does the violation of the sample orthogonality cause bias in this case? If yes, what is the direction of bias? If the sample orthogonality is violated, would using IVs to address endogeneity improve the performance of the convex regression estimator or rather make things worse? In Section \ref{sec:mc}, we shed new light on these questions through the following Monte Carlo simulations and an empirical illustration.  

\subsection{Convex regression subject to concavity, monotonicity, and linear homogeneity}

Let us further assume the regression function $f$ to be concave, monotonic increasing, and linearly homogeneous in both additive and multiplicative models (i.e., problems \eqref{eq:eq3} and \eqref{eq:eq4}). 

While the linear homogeneity is considered in both convex regression problems, the sample orthogonality conditions will not be altered (see Theorem~\ref{theo4}).
\begin{theorem}
        The linear homogeneity does not affect the sample orthogonality condition of convex regression.
    \label{theo4}
\end{theorem}
\begin{proof}
    See Appendix~\ref{sec:proofs}. 
\end{proof}

In the additive model, at the optimal solution to problem \eqref{eq:eq3}, the sum of residuals is not necessarily equal to zero, i.e., $\sum\limits_{i=1}^n e_i \neq 0$ and consequently $\sum\limits_{i=1}^n y_i \neq \sum\limits_{i=1}^n \hat{y}_i$. Unsurprisingly, this is analogous to the OLS regression through the origin (without an intercept), in which the sum of residuals is not usually zero \citep{Eisenhauer2003}. 

Similarly, the sum of residuals is not necessarily zero in the multiplicative model \eqref{eq:eq4}, but the orthogonality condition can be held if input variables are scaled by the fitted values (i.e., $\sum\limits_{i=1}^ne_i\dfrac{x_{ik}}{\hat{y}_i}=0 \,\ \forall k$) in the case of multiplicative model subject linear homogeneity and concave.

The validity of the orthogonality condition is mainly due to the first-order condition of the slope coefficients $\bbeta$, and imposing a sign constraint on the intercept coefficient $\alpha$ does not violate the orthogonality condition. Therefore, imposing a sign constraint on the intercept parameter $\alpha$ does not violate the orthogonality condition.

%-----------------%
%
%-----------------%

\section{Convex regression with instrumental variables and a control function}\label{sec:cf}

Suppose that in the nonparametric additive model~\eqref{eq:eq1}, one of the explanatory variables, $x_1$, is endogenous, meaning either
\begin{equation*}
    \E(\varepsilon | x_1) \neq 0 \text{ or } \E(y | \bx) \neq f(\bx).
\end{equation*}
This implies that $x_1$ is correlated with the random error term $\varepsilon$, or the conditional expectation of $y$ given $\bx$ deviates from the true underlying function $f(\bx)$. That is, the nonparametric estimator $\hat{f}(\bx)$ is inconsistent, $ \hat{f}(\bx) \overset{p}{\centernot\longrightarrow} f(\bx)$.

The classic remedy to such endogeneity bias is to utilize instrumental variables that are highly correlated with the regressors but uncorrelated with the disturbance term. In the present context of convex regression, good instruments $\bz$ should satisfy the relevant population orthogonality conditions identified in the previous section. If such instruments are available, the classic 2SLS approach can be easily adapted to the convex regression as follows. In the first stage, we regress the endogenous regressor $x_e$ on instruments $\bz$ and other exogenous variables using linear regression. In the second stage, we replace the endogenous $x_e$ by the predicted $\hat{x}_e$ obtained in the first stage.

Inspired by \citet{Yatchew2003} and \citet{Chetverikov2017}, we also consider an alternative hybrid IV control function approach to alleviate the endogeneity bias in the convex regression. In this approach, the instruments are introduced by means of a control function and the two-stage estimation strategy is then applied to obtain consistent estimates in the presence of endogeneity. 

Suppose there exist instruments $\bz$ such that
\begin{equation}
    x_{1i} = \bomega \bz_i + \zeta_i,
    \label{eq:eq4.1}
\end{equation}
where $\E(\zeta | \bz) = 0$ and $\E(\varepsilon | \bz) = 0$. For instance, when estimating the regression function with measurement error in inputs, we can instrument the capital stock (i.e., endogenous variable $x_1$) with (lagged) investment (i.e., exogenous variable $z$; an alternative measure of capital) \citep{Collard2016}. 

We notice that in earlier studies, the control function is usually a single variable that, when added to a regression, can reduce the endogeneity of a policy variable, but in modern econometrics, the valid control function generally relies on the availability of one or more instruments \citep{Wooldridge2015}. In this sense, the control function approach inherits the spirit of the IV approach but can overcome the ill-posed inverse problem in nonparametric regression. We thus practically consider an instrument and other exogenous variables as the valid control function, $\bz$, in the empirical illustration. 

Suppose further that $\E(\varepsilon |x_1, \zeta) = \lambda \zeta$ and hence we have $\varepsilon = \lambda \zeta + \nu$. We need to regress \eqref{eq:eq4.1} using OLS to obtain residuals $\zeta$ and then rewrite the nonparametric additive model~\eqref{eq:eq1} to the following semi-nonparametric model\footnote{
    The same control function approach can also be applied to the nonparametric multiplicative model~\eqref{eq:eq4} to address the endogeneity problem.
}
\begin{equation}
   y_i = f(\bx_i) + \lambda \hat{\zeta}_i + \nu_i,
    \label{eq:eq4.2}
\end{equation}
where $\E(\nu | \bx, \hat{\zeta}) = 0$ and $\text{Var}(\nu | \bx, \hat{\zeta})=\sigma^2(\bx, \hat{\zeta})$. Such a two-stage estimation to correct endogeneity is also sometimes called a residual inclusion method in the literature \citep[see, e.g.,][]{Terza2008, Amsler2016}. Compared to the 2SLS approach, the hybrid IV control function approach can not only capture the impact of measurement error in the input but also control the unobserved simultaneous bias.  

To estimate $\lambda$ and $f$ in~\eqref{eq:eq4.2}, we can solve the following semi-nonparametric convex regression model
\begin{alignat}{2}
    \underset{\alpha, \bbeta, \lambda, \nu}\min \, &\dfrac{1}{2}\sum_{i=1}^n\nu_i^2 &{\quad}& \label{eq:eq4.3}\\
    \mbox{\textit{s.t.}}\quad 
    &  y_i= \alpha_i+\bbeta^\prime_i\bx_i + \lambda \hat{\zeta}_i + \nu_i &{\quad}& \forall i  \notag \\
    &  \alpha_i + \bbeta_i^\prime\bx_i \le \alpha_h + \bbeta_h^\prime\bx_i  &{\quad}&  \forall i, h \notag 
\end{alignat}
where $\hat{\zeta}_i$ is the estimated residuals by the OLS regression~\eqref{eq:eq4.1}. We consider the concave case of the regression function in the additive semi-nonparametric convex regression model \eqref{eq:eq4.3}. The monotonicity constraints are eliminated to maintain the sample orthogonality condition, as suggested by Theorems \ref{theo1} and \ref{theo3}. Similar to convex regression~\eqref{eq:eq3}, the first set of constraints in~\eqref{eq:eq4.3} is the reformulation of semi-nonparametric regression \eqref{eq:eq4.2}, and the second set of constraints guarantees the regression function $f$ to be concave. 

The estimated $\hat{\alpha}_i$ and $\hat{\bbeta}_i$ in~\eqref{eq:eq4.3} vary for each observation, but $\hat{\lambda}$ is common to all observations. Further, the semi-nonparametric convex regression model~\eqref{eq:eq4.3} is a restricted special case of convex regression \eqref{eq:eq3}, in which $\hat{\zeta} \subseteq \bx$ with $\lambda_i = \lambda_h$. The estimated $\hat{\lambda}$ in~\eqref{eq:eq4.3} has been proved consistent, unbiased, asymptotically efficient, and converges at the rate of $O(n^{-1/2})$ \citep{Johnson2011, Johnson2012a}. 

Based on the OLS regression and convex regression, the consistency of the nonparametric regression function $f$ in the proposed hybrid IV control function approach can be stated as
\begin{theorem}
  If the instruments $\bz$ satisfies the conditions $\E(\zeta | \bz) = 0$ and $\E(\varepsilon | \bz) = 0$, then the semi-nonparametric convex regression estimator $\hat{f}(\bx, \hat{\zeta})$ subject to concavity~\eqref{eq:eq4.3} is consistent, that is,
   \begin{equation*}
       \hat{f}(\bx, \hat{\zeta}) \overset{p}{\rightarrow} f(\bx, \hat{\zeta}).
   \end{equation*}
    \label{theo5}
\end{theorem}
\vspace{-3em}
\begin{proof}
    See Appendix~\ref{sec:proofs}.  
\end{proof}

Alternatively, we can utilize the double residual method to estimate $f$ and $\lambda$ separately \citep{Robinson1988}, or even resort to the first-order differencing approach in the univariate case \citep{Yatchew2003}. In addition to the parametric OLS regression, in the first stage~\eqref{eq:eq4.1}, we can also use nonparametric methods such as kernel regression and local polynomial method to estimate $\zeta$ consistently. By contrast to our linear control function, \citet{Rodseth2025} propose a fully nonparametric approach to model the control function. Nevertheless, further comparison of linear and nonparametric modeling of control functions is warranted for future investigations.

Another appealing feature of the proposed hybrid IV control function approach is that we can simply use the $t$- or $F$-test to examine whether $x_1$ is exogenous after the semi-nonparametric convex regression \eqref{eq:eq4.3} is applied (i.e., the null hypothesis is $H_0: \hat{\lambda} = 0$). That is, if $\hat{\lambda}$ significantly departs from zero, then the explanatory variable $x_1$ is endogenous.

%-----------------%
%
%-----------------%

\section{Monte Carlo study}\label{sec:mc}

In this section, we perform a Monte Carlo study to compare the finite sample performance of convex regression and OLS under endogeneity. We have three objectives in designing the following simulations. First, we investigate how large the endogeneity bias of OLS is compared to the functional form-related specification error. Second, we explore whether the 2SLS and IV control function approaches can help convex regression without monotonicity to mitigate the impact of endogeneity in various scenarios. Third, we examine whether the 2SLS and IV control function approaches can improve the performance of convex regression with monotonicity under endogeneity. 

We consider the following two data-generating processes (DGPs) to generate inputs $\bx$ and output $y$ (see, e.g., \citealp{Cordero2015, Rodseth2025}).
\begin{equation*}
\begin{aligned}
    & \text{DGP I:}  \quad  y_i = x_{1i}^{0.4} \cdot x_{ei}^{0.6} \cdot \exp(\varepsilon_i), \\
    & \text{DGP II:} \quad  y_i = x_{1i}^{0.3} \cdot x_{2i}^{0.35} \cdot x_{ei}^{0.35} \cdot \exp(\varepsilon_i),
\end{aligned}
\end{equation*}
where the exogenous $x_1$ and/or $x_2$ are independently drawn from a uniform distribution, $x_1/x_2 \sim U[5, 50]$, while the endogenous $x_e$ and random error $\varepsilon$ are generated according to the specific setting in each of the experiments below. 

In all experiments, we estimate the OLS regression using Python/scikit-learn and convex regression using Python/pyStoNED \citep{Dai2024} with the off-the-shelf solver KNITRO (13.2). Each scenario is duplicated 1000 times to compute the root mean square error (RMSE) and bias statistics for all estimators in terms of the production function estimation. All simulations are run on Finland's high-performance computing cluster Puhti with Xeon @2.2 GHz processors, 2 CPUs, and 20 GB of RAM per task.

\subsection{Experiment 1}

The endogeneity problem of OLS is well understood: if the basic assumption $\E(\varepsilon_i x_{ik})=0, \exists k$ fails, then the standard OLS estimator is biased and inconsistent. However, there is little simulation evidence on how large the endogeneity bias is compared to the specification error, for example, due to the mis-specified parametric functional form.\footnote{
    The true model would be in logs, so the misspecification is due to the fact that the OLS regression is in levels. In the first experiment, both OLS regressions---specified in logarithmic and level forms---are employed to assess whether they are subject to endogeneity.
} 
Therefore, we design the first experiment to examine the finite sample performance of OLS and convex regression estimators under endogeneity.

Similar to \cite{Mutter2013}, we generate the random error $\varepsilon$ as 
\begin{equation*}
    \varepsilon = \sigma_\varepsilon \left[\rho \times \text{std}(x_e) + (1-\rho^2)^{.5}W\right],
\end{equation*}
where the parameter $\rho \in [0, 1]$ is a prespecified correlation coefficient determining the degree of endogeneity between $x_e$ and $\varepsilon$, the random variable $W \sim N(0, 1)$, and $\text{std}(x_e)$ is the standardized endogenous variable $x_e$ with a mean of zero and a standard deviation of one. The endogenous $x_e$ is independently drawn from a uniform distribution, $x_e \sim U[5, 50]$.

In experiment 1, we consider 72 scenarios with different numbers of observations $n \in \{50, 100, 200, 400\}$, different magnitudes of endogeneity $\rho \in \{0, 0.45, 0.9\}$, different noise levels $\sigma_\varepsilon \in \{0.5, 1, 2\}$, and different numbers of inputs $k \in \{2, 3\}$. The standard OLS regression and multiplicative convex regression~\eqref{eq:eq4} are applied to the designed DGPs for estimating production functions. The RMSE statistics of the alternative estimators are reported in Tables~\ref{tab:tab1} (the bias statistics are available in Table~\ref{tab:a1} in the Appendix). 
\begin{table}[H]
  \centering
  \caption{RMSE comparison between OLS and convex regression with $n=400$ (standard deviations in parentheses).}
  \footnotesize
    \setlength{\tabcolsep}{4.5pt}
    \begin{tabular}{llccccccccc}
    \toprule
    \multicolumn{1}{l}{\multirow{3}[6]{*}{\begin{tabular}{l} Endogeneity \\ $\rho$ \end{tabular}}}
 & \multicolumn{1}{l}{\multirow{3}[6]{*}{\begin{tabular}{l} Noise \\ $\sigma_\varepsilon$ \end{tabular}}} & \multicolumn{4}{c}{DGP I: $k=2$} &       & \multicolumn{4}{c}{DGP II: $k=3$} \\
\cmidrule{3-6}\cmidrule{8-11}          &       & \multicolumn{2}{c}{OLS} & \multicolumn{2}{c}{CR} &       & \multicolumn{2}{c}{OLS} & \multicolumn{2}{c}{CR} \\
\cmidrule{3-6}\cmidrule{8-11}          &       & \multicolumn{1}{c}{logs} & \multicolumn{1}{c}{levels} & \multicolumn{1}{c}{conc+mon} & \multicolumn{1}{c}{conc} &       & \multicolumn{1}{c}{logs} & \multicolumn{1}{c}{levels} & \multicolumn{1}{c}{conc+mon} & \multicolumn{1}{c}{conc} \\
    \midrule
    0     & 0.5   & 0.98  & 4.31  & 1.69  & 2.33  &       & 1.13  & 4.23  & 2.04  & 3.35 \\
          &       & (0.49)  & (0.86)  & (0.46)  & (0.52)  &       & (0.44)  & (0.77)  & (0.34)  & (0.44) \\
          & 1     & 1.96  & 18.04 & 2.81  & 3.93  &       & 2.26  & 17.74 & 3.18  & 5.57 \\
          &       & (1.01)  & (3.66)  & (1.17)  & (0.93)  &       & (0.91)  & (3.35)  & (0.66)  & (0.84) \\
          & 2     & 3.97  & 181.04 & 4.52  & 6.41  &       & 4.24  & 169.92 & 4.52  & 8.22 \\
          &       & (2.14)  & (87.36) & (1.64)  & (1.77)  &       & (2.27)  & (101.33) & (1.95)  & (3.05) \\
    0.9   & 0.5   & 12.06 & 15.59 & 7.62  & 7.64  &       & 11.29 & 14.34 & 8.90  & 9.02 \\
          &       & (1.92)  & (2.48)  & (1.23)  & (1.23)  &       & (0.56)  & (0.71)  & (0.43)  & (0.43) \\
          & 1     & 29.55 & 44.93 & 11.36 & 11.51 &       & 27.59 & 41.49 & 14.45 & 15.11 \\
          &       & (6.88)  & (10.56) & (2.69)  & (2.72)  &       & (1.83)  & (3.03)  & (0.93)  & (0.94) \\
          & 2     & 106.14 & 265.20 & 14.50 & 15.23 &       & 78.68 & 195.16 & 15.46 & 17.86 \\
          &       & (20.82) & (58.91) & (2.93)  & (3.05)  &       & (37.32) & (95.22) & (7.32)  & (8.45) \\
    \bottomrule
    \\[-0.8em]
    \multicolumn{11}{l}{\footnotesize \textit{Notes}: logs: OLS regressions in logarithmic form; levels: OLS regressions in level form;} \\
    \multicolumn{11}{l}{\footnotesize \hspace*{1cm} conc+mono: the multiplicative model subject to monotonicity and concavity;}\\
    \multicolumn{11}{l}{\footnotesize \hspace*{1cm} conc: the multiplicative model subject to concavity.}\\
    \end{tabular}%
  \label{tab:tab1}%
  \vspace{-1em}
\end{table}%

The correctly specified OLS regression (logs) yields the smallest RMSE in the case of no endogeneity in all considered scenarios and hence outperforms the multiplicative convex regression, which conforms with the findings in, e.g., \citet{Kuosmanen2008} and \citet{Tsionas2022b}. However, as the degree of endogeneity increases (i.e., $\rho \rightarrow 1$), the performance of the OLS regression (logs) 
 quickly deteriorates (see Table~\ref{tab:tab2} for $\rho=0.45$), and the RMSE of the OLS regression (logs) is notably higher than that of convex regression, irrespective of monotonicity constraints, particularly when the error standard deviation ($\sigma_\varepsilon$) is high. Yet, the performance of convex regression is relatively robust to the changes in endogeneity level. This suggests that the impact of endogeneity is more pronounced for the OLS than the convex regression estimator. Note that the misspecified OLS regression (levels) exhibits the highest RMSE across all scenarios, implying that the biases due to model misspecification outweigh those stemming from endogeneity.

When comparing the difference in RMSE between OLS and convex regression, we observe that the gap narrows in the absence of endogeneity but widens largely when endogeneity is introduced to DGPs. This observation remains unchanged across all noise levels.
For instance, in the scenarios with $\sigma_\varepsilon=1$ and $k=2$, the absolute difference in RMSE between OLS (logs) and convex regression (conc) is 1.97, 4.43, and 18.04 with three given $\rho$, respectively (see Table~\ref{tab:tab2} for $\rho=0.45$). Furthermore, the RMSE results indicate that, even when alternative functional forms are considered (e.g., OLS in logs), OLS estimators still perform notably worse than convex regression methods, especially when $\rho$ is large. This suggests that, under high endogeneity, the bias from endogeneity in OLS may outweigh the bias arising from functional form misspecification.

Across all DGPs and under varying levels of noise and endogeneity, the convex regression model imposing both concavity and monotonicity constraints (conc+mon) systematically outperforms its counterpart that imposes only concavity (conc) in terms of RMSE. The performance gains from incorporating monotonicity are particularly pronounced when the variance of the random error term is large, highlighting the usefulness of monotonicity constraints. Even under conditions of severe endogeneity, the conc+mon specification yields marginally lower RMSEs, suggesting that the additional structural constraint enhances the estimator's robustness without introducing substantial bias or overfitting (cf., Table~\ref{tab:a1}).

Several additional insights emerge from Tables~\ref{tab:tab1} and~\ref{tab:a1}. First, the bias associated with OLS (logs) is always positive, as expected, due to the non-linear nature of the logarithmic transformation. Second, for both OLS and convex regression estimators, both RMSE and bias increase with the level of noise, indicating a deterioration in estimation accuracy as the signal-to-noise ratio worsens. Third, as the dimensionality of the input space increases, the RMSE of convex regression estimators tends to rise, supporting the findings from prior work such as \citet{Dai2023} and \citet{Dai2023c}, which highlight the curse of dimensionality of nonparametric estimators. Finally, in the presence of endogeneity (e.g., $\rho = 0.9$), the bias of convex regression estimators remains substantially lower than that of OLS, further reinforcing the result that convex regression exhibits greater robustness to endogeneity bias.
\begin{table}[H]
  \centering
  \caption{The effect of sample size with moderate endogeneity with $\rho=0.45$ (standard deviations in parentheses).}
  \footnotesize
    \setlength{\tabcolsep}{4.5pt}
    \begin{tabular}{lcccccccccc}
    \toprule
    \multicolumn{1}{l}{\multirow{3}[6]{*}{\begin{tabular}{l} Noise \\ $\sigma_\varepsilon$ \end{tabular}}} & \multicolumn{1}{l}{\multirow{3}[6]{*}{\begin{tabular}{l} Sample size \\ $n$ \end{tabular}}} & \multicolumn{4}{c}{DGP I: $k=2$} &       & \multicolumn{4}{c}{DGP II: $k=3$} \\
\cmidrule{3-6}\cmidrule{8-11}          &       & \multicolumn{2}{c}{OLS} & \multicolumn{2}{c}{CR} &       & \multicolumn{2}{c}{OLS} & \multicolumn{2}{c}{CR} \\
\cmidrule{3-6}\cmidrule{8-11}          &       & \multicolumn{1}{c}{logs} & \multicolumn{1}{c}{levels} & \multicolumn{1}{c}{conc+mon} & \multicolumn{1}{c}{conc} &       & \multicolumn{1}{c}{logs} & \multicolumn{1}{c}{levels} & \multicolumn{1}{c}{conc+mon} & \multicolumn{1}{c}{conc} \\
    \midrule
    0.5   & 50    & 6.19  & 9.85  & 5.22  & 5.90  &       & 6.19  & 9.18  & 6.19  & 7.94 \\
          &       & (2.35)  & (3.42)  & (1.78)  & (1.81)  &       & (1.97)  & (2.86)  & (1.64)  & (1.83) \\
          & 100   & 5.93  & 9.77  & 4.81  & 5.21  &       & 5.71  & 8.86  & 5.51  & 6.67 \\
          &       & (1.70)  & (2.43)  & (1.31)  & (1.31)  &       & (1.45)  & (2.19)  & (1.19)  & (1.24) \\
          & 200   & 5.78  & 9.73  & 4.61  & 4.81  &       & 5.52  & 8.71  & 5.14  & 5.85 \\
          &       & (1.18)  & (1.71)  & (0.93)  & (0.92)  &       & (1.03)  & (1.51)  & (0.86)  & (0.87) \\
          & 400   & 5.67  & 9.59  & 4.47  & 4.58  &       & 5.38  & 8.63  & 4.91  & 5.35 \\
          &       & (0.85)  & (1.19)  & (0.65)  & (0.63)  &       & (0.73)  & (1.04)  & (0.61)  & (0.60) \\
    1     & 50    & 13.98 & 33.78 & 9.33  & 10.79 &       & 13.84 & 31.00 & 11.12 & 15.42 \\
          &       & (6.31)  & (15.54) & (3.92)  & (4.10)  &       & (5.45)  & (13.32) & (3.97)  & (4.80) \\
          & 100   & 13.22 & 33.81 & 8.55  & 9.33  &       & 12.46 & 30.41 & 9.87  & 12.43 \\
          &       & (4.45)  & (11.16) & (2.80)  & (2.86)  &       & (3.90)  & (10.50) & (2.70)  & (3.01) \\
          & 200   & 12.79 & 33.59 & 8.17  & 8.52  &       & 11.95 & 29.64 & 9.25  & 10.71 \\
          &       & (2.98)  & (7.28)  & (1.92)  & (1.92)  &       & (2.69)  & (6.81)  & (1.86)  & (1.96) \\
          & 400   & 12.41 & 32.82 & 7.81  & 7.98  &       & 11.56 & 29.42 & 8.84  & 9.67 \\
          &       & (2.33)  & (5.56)  & (1.44)  & (1.42)  &       & (1.88)  & (4.57)  & (1.29)  & (1.31) \\
    2     & 50    & 38.65 & 306.93 & 16.01 & 19.47 &       & 39.10 & 286.56 & 20.15 & 32.12 \\
          &       & (23.79) & (464.69) & (9.39)  & (10.36) &       & (22.17) & (375.68) & (10.56) & (15.26) \\
          & 100   & 35.08 & 320.12 & 13.81 & 15.69 &       & 32.81 & 296.03 & 16.84 & 23.33 \\
          &       & (15.58) & (342.62) & (6.19)  & (6.67)  &       & (14.30) & (388.38) & (6.44)  & (8.28) \\
          & 200   & 32.93 & 303.21 & 12.73 & 13.64 &       & 30.28 & 269.21 & 15.21 & 18.72 \\
          &       & (9.98)  & (186.61) & (4.11)  & (4.25)  &       & (9.40)  & (218.54) & (4.21)  & (4.91) \\
          & 400   & 31.26 & 287.38 & 11.84 & 12.30 &       & 28.49 & 261.78 & 14.18 & 16.11 \\
          &       & (7.97)  & (130.03) & (3.09)  & (3.12)  &       & (6.74)  & (118.87) & (3.05)  & (3.32) \\
    \bottomrule
    \\[-0.8em]
    \multicolumn{11}{l}{\footnotesize \textit{Notes}: logs: OLS regressions in logarithmic form; levels: OLS regressions in level form;} \\
    \multicolumn{11}{l}{\footnotesize \hspace*{1cm} conc+mono: the multiplicative model subject to monotonicity and concavity;}\\
    \multicolumn{11}{l}{\footnotesize \hspace*{1cm} conc: the multiplicative model subject to concavity.}\\
    \end{tabular}%
  \label{tab:tab2}%
  \vspace{-1em}
\end{table}%

We further investigate the performance of both estimators for different sample sizes and for a wide range of variances of random error $\sigma_\varepsilon$ in Table~\ref{tab:tab2}. In the case of moderate endogeneity, we observe that the larger $n$ for either $k=2$ or $k=3$, the better performance in terms of RMSE and bias for all methods. The average standard deviations of RMSE clearly decrease as $n$ increases. Again, the results in Table~\ref{tab:tab2} demonstrate that as $\sigma_\varepsilon$ decreases, the performance of the estimator gets better. 

\subsection{Experiment 2}

As mentioned in Section~\ref{orth1}, convex regression subject to concavity can satisfy the sample orthogonality condition. A valid instrument is then expected to improve the performance of such a convex regression estimator in the presence of endogeneity. We thus design the second experiment to compare the finite sample performance of convex regression and convex regression with the 2SLS and IV control function approaches, where the monotonicity constraint in all estimators is relaxed.

We draw the endogenous $x_e$ from the following multiplicative model
$$x_{ei} = \tilde{\rho} x_{zi} \cdot \exp(\nu_i),$$ 
where $x_z$ is an instrument to be used in the control function and 2SLS approaches, and the parameter $\tilde{\rho}$ captures the strength of correlation between $x_z$ and $x_e$. For the individual $\varepsilon_i$ and $\nu_i$, we assume 
$$\left(\begin{array}{c} \varepsilon_i \\ \nu_i \end{array}\right) \sim N(0, \Sigma); \Sigma = \left(\begin{array}{cc} \sigma_\varepsilon^2  & \sigma_{\varepsilon\nu}\\ \sigma_{\varepsilon\nu} & \sigma_\nu^2 \end{array}\right).$$

Since $\text{Cov}(\varepsilon, \nu) \neq 0$, it follows that $\text{Cov}(x_e, \varepsilon) \neq 0$, confirming that $x_e$ is an endogenous variable. Given that $\text{Cov}(x_z, \nu) = 0$, $x_z$ can serve as a valid instrument  for $x_e$ and is randomly generated from a uniform distribution, $x_z \sim U[5, 50]$. We further set $\sigma_\varepsilon = \sigma_\nu = \sigma$ uniformly across all considered cases. In the simulations, we choose $n \in \{50, 100, 200, 400\}$, $\rho \in \{0, 0.45, 0.9\}$, $\tilde{\rho} \in \{0.45, 0.9\}$, $\sigma \in \{1, 2, 5\}$ and $k \in \{2, 3\}$. The 2SLS approach and the hybrid IV control function approach introduced in Section~\ref{sec:cf} are applied.

Fig.~\ref{fig:fig1} depicts the RMSE statistic of three different convex regression estimators with $\sigma=2$ and $k=3$. See also the scenarios with $\sigma=1$ and $k=3$ in Fig.~\ref{fig:figa1}. From the demonstrated figures, we observe that in the case of no endogeneity ($\rho=0$), convex regression generally performs best in three estimators, but the performance difference between convex regression and convex regression with IV control function is quite small, particularly in the large sample (e.g., $n=400$). Regarding the IV-based convex regression, its performance deteriorates significantly when the sample size exceeds 100.
\begin{figure}[H]
    \centering
    \includegraphics[width=\linewidth]{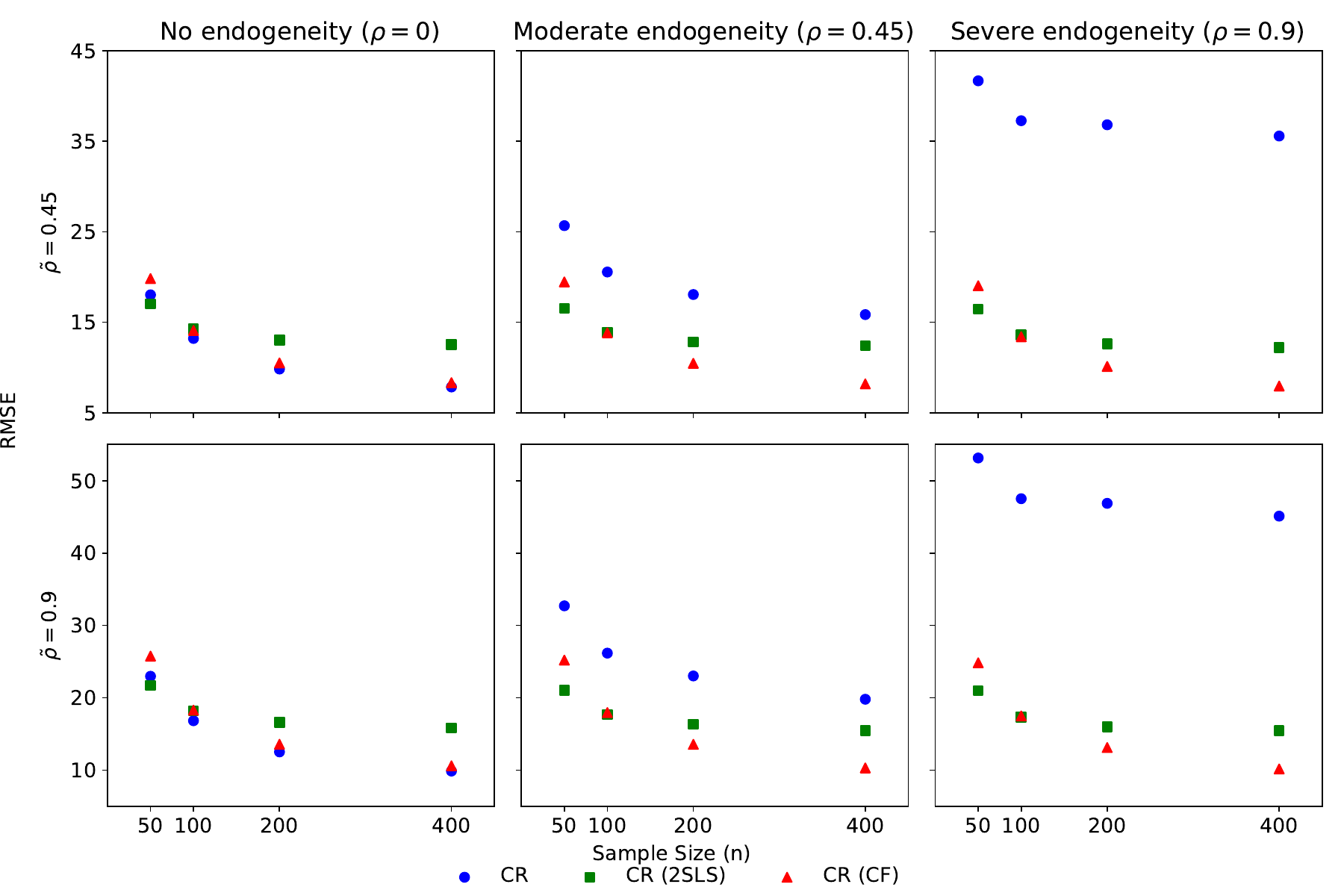}
    \caption{The RMSE statistic of convex regression estimators with $\sigma=2$ and $k=3$.}
    \label{fig:fig1}
\end{figure}
After introducing a correlation between the input and the random error (i.e., $\rho>0$) to the data space, the RMSE results of convex regression clearly indicate the presence of an endogeneity bias, similar to the findings in Tables~\ref{tab:tab1} and \ref{tab:tab2}. While convex regression subject to concavity can satisfy the sample orthogonality condition, it is also not immune to endogeneity bias. As shown in Fig.~\ref{fig:fig1}, the IV-based and IV control function-based convex regression estimators can both be used to address the potential endogeneity problem, but the latter is more effective. Specifically, the RMSEs of convex regression with IV and IV control function approach are less than those of convex regression in the cases of $\rho=0.45$ and $\rho=0.9$, and the RMSE of IV-based convex regression is also higher than that of IV control convex regression. Furthermore, the strong instrument does not appear to reduce RMSE values significantly in both IV-related approaches. 

We next evaluate the performance of estimators from the bias perspective. Overall, the bias results presented in Fig.~\ref{fig:fig2} align with the RMSE findings shown in Fig.~\ref{fig:fig1}. In scenarios where $n=200$ or $n=400$, the bias of the convex regression with the IV approach is negative, while the biases of both convex regression and convex regression with the IV control function approach are positive across all considered scenarios, as expected. However, bias should always be positive in the designed multiplicative error structure. Furthermore, for the large sample size, the bias of the IV control function approach is closer to zero compared to the 2SLS approach. Therefore, we would stress that the proposed hybrid IV control function approach is more suitable to address the endogeneity problem in the context of nonparametric convex regression. 
\begin{figure}[H]
    \centering
    \includegraphics[width=\linewidth]{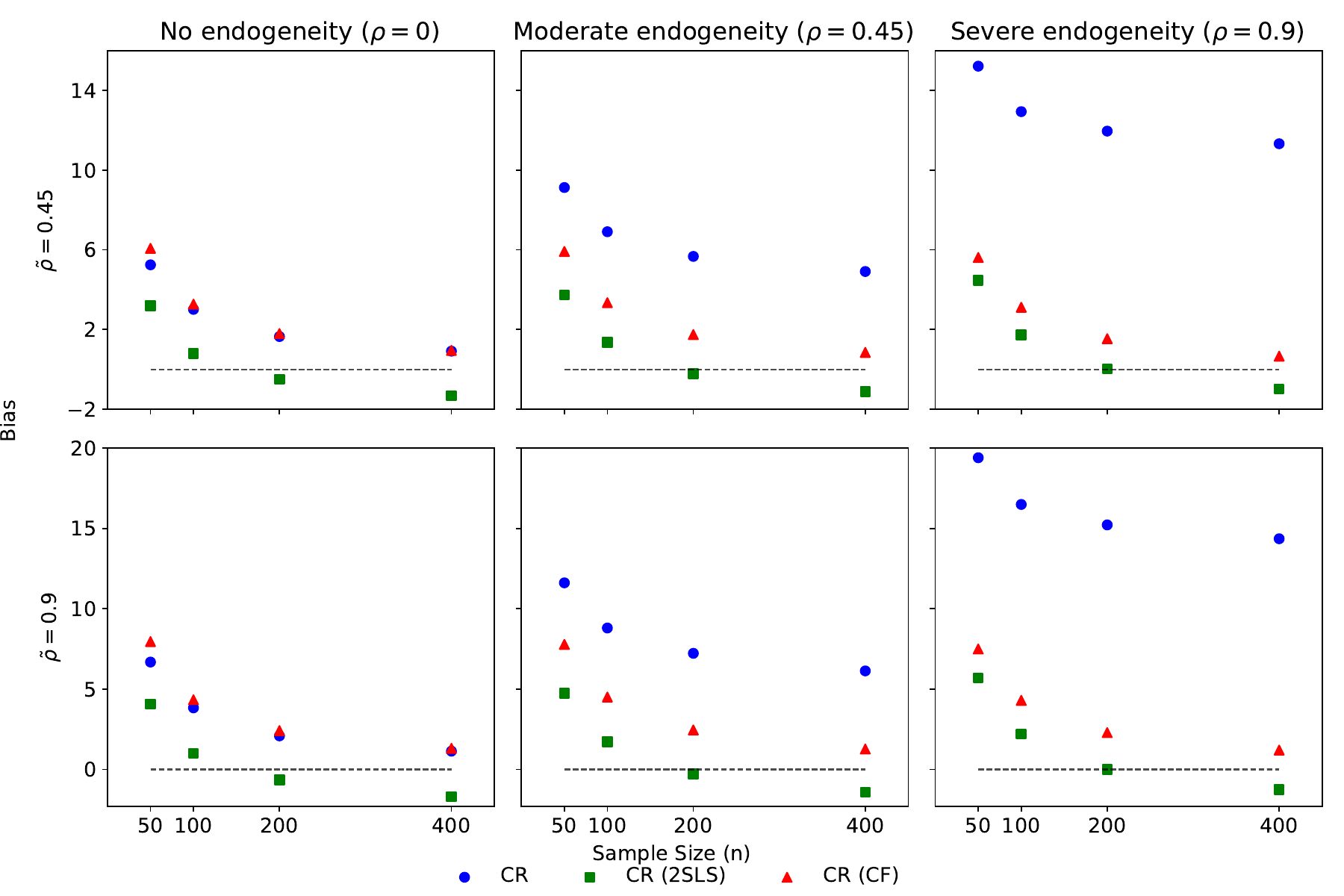}
    \caption{The bias statistic of convex regression estimators with $\sigma=2$ and $k=3$.}
    \label{fig:fig2}
\end{figure}

\subsection{Experiment 3}

Section~\ref{orth2} shows that the sample orthogonality condition cannot hold in the case of monotonicity, and the exact population orthogonality condition is also unknown. We thus design experiment 3 to test if the performance of convex regression with monotonicity improves or deteriorates when the IV and IV control function approaches are introduced. The same DGPs and scenario settings as Experiment 2 are considered. 
\begin{figure}[H]
    \centering
    \includegraphics[width=\linewidth]{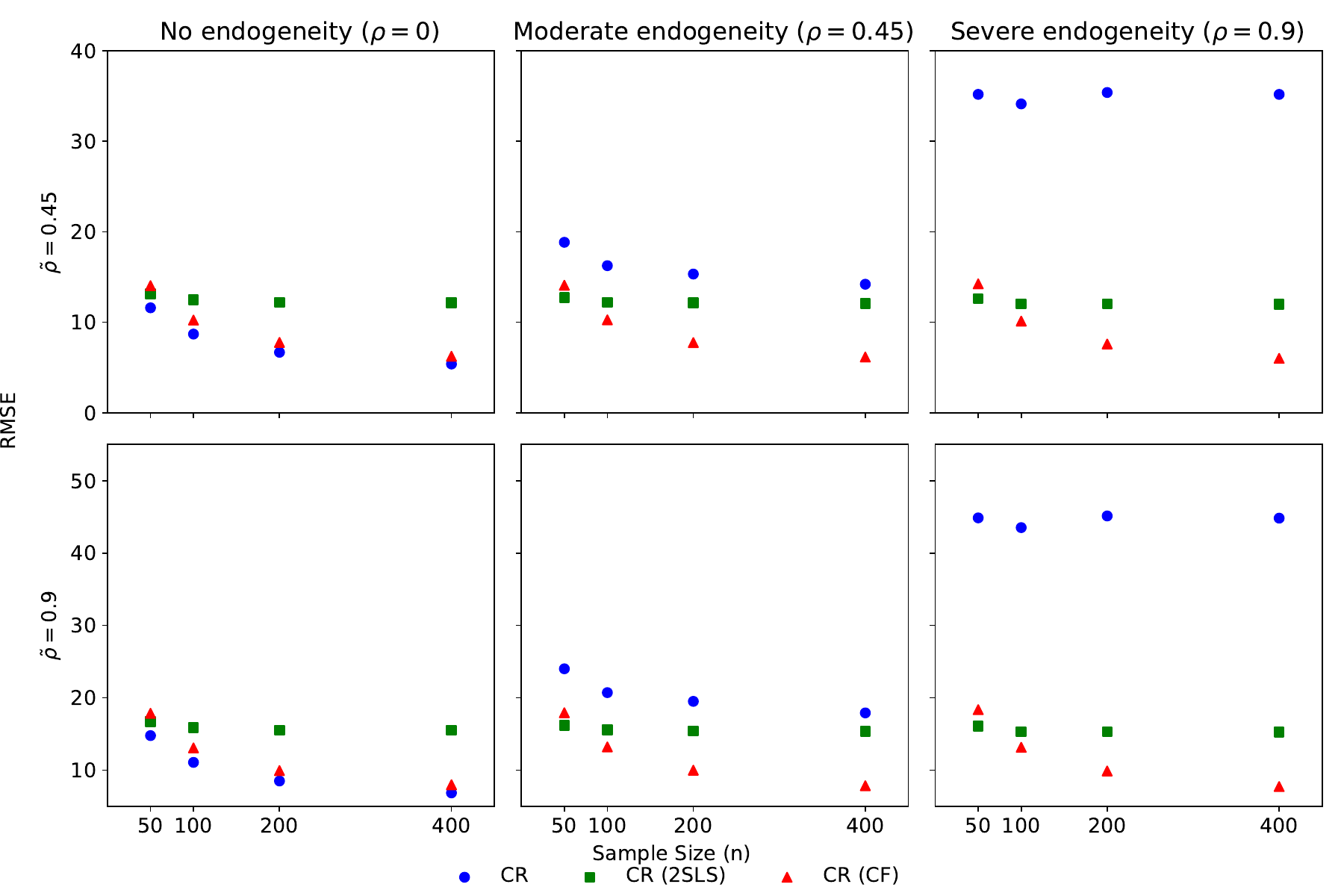}
    \caption{The RMSE statistic of convex regression estimators in the case of monotonicity with $\sigma_\varepsilon=2$ and $k=3$.}
    \label{fig:fig3}
\end{figure}

Fig.~\ref{fig:fig3} plots the RMSE statistic of three convex regression estimators in the case of monotonicity for the low noise and three-input setting (see Fig.~\ref{fig:figa2} for $k=2$). When the correlation coefficient $\rho=0$, convex regression outperforms the other two IV-based convex regression methods, as indicated by its lower RMSE. When there is endogeneity between $x_e$ and $\varepsilon$ in the DGPs, the IV control function convex regression produces the smallest RMSEs while accounting for monotonicity. The 2SLS-based convex regression performs better than convex regression in cases of severe endogeneity in $x_e$ ($\rho=0.9$) or more inputs (Fig.~\ref{fig:figa2}). Overall, while the sample orthogonality condition cannot be satisfied in the case of monotonicity, the IV-based or the IV control function-based convex regression can, to some extent, help to address the endogeneity problem, but the performance of the IV control function approach is more robust.

A further comparison of Figs.~\ref{fig:fig1} and \ref{fig:fig3} provides evidence that imposing monotonicity constraints in convex regression slightly decreases RMSEs and improves its performance across all scenarios, irrespective of whether IV or IV control function approaches are applied. This finding persists across various degrees of endogeneity and sample sizes, and also confirms the result discussed in Table~\ref{tab:tab1}. While imposing monotonicity constraints in convex regression theoretically violates the sample orthogonality condition, it often leads to improved performance in terms of RMSE. The better performance can be explained by several factors. When the DGP follows a monotonic and convex form, such as the Cobb-Douglas function, imposing monotonicity constraints brings the estimator closer to the true structure and reduces misspecification risk. In finite samples, the bias from violating orthogonality is often limited, while the variance reduction from the constraint can lead to better overall accuracy.

%-----------------%
%
%-----------------%

\section{Application}\label{sec:appl}

The existing studies typically rely on a control function to address endogeneity concerns in estimating the production function \citep[see, e.g.,][]{Olley1996, Levinsohn2003, Ackerberg2015}. It is thus worth examining whether the hybrid IV control function approach empirically makes a difference in convex regression. Furthermore, excluding the monotonicity assumption from convex regression can guarantee its sample orthogonality conditions. It is also interesting to investigate whether monotonicity is necessary in the endogeneity correction.

To explore the above two empirical objectives, we apply the proposed hybrid IV control function to the Chilean manufacturing production data, covering all firms with more than 10 employees from 1979--1996. Such census data conducted by Chile's National Institute of Statistics has been widely used in works such as \citet{Levinsohn2003}, \citet{Gandhi2020}, and \citet{Tsionas2022b}. 

We restrict our estimation sample to a subset of five manufacturing industries in the years 1995 and 1996, with ISIC codes 311 (Food Products), 321 (Textiles), 322 (Apparel), 331 (Wood Products), and 381 (Metals). For each industry, the sample data include the real gross output (measured as deflated revenues), labor (a weighted sum of blue- and white-collar workers), capital (measured by the perpetual inventory method), intermediate inputs (a sum of raw materials, energy, and service expenditures), and investment. A more detailed description of the data and the selected variables is available in \citet{Gandhi2020}. 

We start by testing if the endogeneity problem exists in production function estimation with Chilean manufacturing data. In doing so, we use the lagged investment as the instrument for capital input in the hybrid IV control function approach. Given the potential heterogeneity among manufacturing industries, we conduct separate estimations for each industry. The endogeneity test results in Table \ref{tab:tab3} indicate that capital input is an endogenous variable for all industries, at least at the 90\% significance level, regardless of whether monotonicity is relaxed. To correct the endogeneity bias in production function estimation, we need the hybrid IV control function approach. 
\begin{table}[H]
  \centering
  \caption{Endogeneity test for capital input.}
    \begin{tabular}{lccccc}
    \toprule
     $H_0: \hat{\lambda} = 0$ & 311   & 321   & 322   & 331   & 381 \\
    \midrule
    Without monotonicity & -0.065\sym{***} & -0.245\sym{***} & -0.036\sym{*}   & -0.108\sym{***} & -0.151\sym{***} \\
    With monotonicity  & -0.082\sym{***}   & -0.217\sym{***} & -0.070\sym{***} & -0.101\sym{***} & -0.153\sym{***} \\
    \bottomrule
    \\[-0.8em]
     \multicolumn{6}{l}{\footnotesize \sym{***} \(p<0.01\); \sym{*} \(p<0.1\).}\\
    \end{tabular}%
    \vspace{-1em}
  \label{tab:tab3}%
\end{table}%

We next report the median coefficients across observations for each industry in Table~\ref{tab:tab4}, where the 25th and 75th percentiles of the coefficients are shown in parentheses. For a thorough comparison, we discuss the estimates by the convex regression with monotonicity, the 2SLS-based convex regression with monotonicity, and the IV control function-based convex regression with and without monotonicity in Table~\ref{tab:tab4}. 
\begin{sidewaystable}
\begin{table}[H]\centering
\caption{Estimates of the production function.}
\fontsize{9.5pt}{11pt}\selectfont 
    \begin{tabular}{rccccccccc}
    \hline
    \multicolumn{1}{l}{\multirow{2}[4]{*}{Estimator}} & \multicolumn{3}{c}{Food Products (311)} & \multicolumn{3}{c}{Textiles (321)} & \multicolumn{3}{c}{Apparel (322)} \\
    \cmidrule{2-10}     & L     & K     & M     & L     & K     & M     & L & K & M \\
    \hline
    \multicolumn{1}{l}{CR} & 0.029 & 0.004 & 0.012 & 0.105 & 0.001 & 0.011 & 0.031 & 0.008 & 0.012 \\
          & [0.014, 0.039] & [0.002, 0.006] & [0.011, 0.013] & [0.013, 0.158] & [0.000, 0.003] & [0.008, 0.014] & [0.000, 0.067] & [0.000, 0.015] & [0.005, 0.016] \\
          \addlinespace
    \multicolumn{1}{l}{CR(IV)} & 0.017 & 0.012 & 0.008 & 0.084 & 0.012 & 0.004 & 0.035 & 0.011 & 0.008 \\
          & [0.012, 0.033] & [0.011, 0.012] & [0.005, 0.009] & [0.005, 0.151] & [0.005, 0.013] & [0.000, 0.007] & [0.001, 0.082] & [0.010, 0.015] & [0.000, 0.011] \\
          \addlinespace
    \multicolumn{1}{l}{CR1(CF)} & 0.025 & 0.006 & 0.012 & 0.075 & 0.007 & 0.010 & 0.031 & 0.011 & 0.011 \\
          & [0.010, 0.032] & [0.004, 0.009] & [0.010, 0.012] & [0.003, 0.116] & [0.002, 0.011] & [0.004, 0.012] & [0.000, 0.071] & [0.000, 0.018] & [0.005, 0.016] \\
          \addlinespace
    \multicolumn{1}{l}{CR2(CF)} & 0.027 & 0.006 & 0.012 & 0.071 & 0.009 & 0.010 & 0.052 & 0.012 & 0.013 \\
          & [0.010, 0.033] & [0.004, 0.008] & [0.011, 0.012] & [-0.050, 0.116] & [0.003, 0.016] & [0.006, 0.014] &[-0.059, 0.079] & [-0.001, 0.021] & [0.008, 0.020] \\
    \hline
    \addlinespace
    \hline
    \multicolumn{1}{l}{\multirow{2}[4]{*}{Estimator}} & \multicolumn{3}{c}{Wood Products (331)} & \multicolumn{3}{c}{Fabricated Metals (381)} &       &       &  \\
\cmidrule{2-7}          & L     & K     & M     & L     & K     & M     &       &       &  \\
\cmidrule{1-7}    \multicolumn{1}{l}{CR} & 0.022 & 0.002 & 0.011 & 0.139 & 0.005 & 0.016 &       &       &  \\
          & [0.002, 0.015] & [0.001, 0.003] & [0.009, 0.011] & [0.016, 0.180] & [0.001, 0.006] & [0.013, 0.018] &       &       &  \\
          \addlinespace
    \multicolumn{1}{l}{CR(IV)} & 0.013 & 0.010 & 0.003 & 0.130 & 0.014 & 0.009 &       &       &  \\
          & [0.003, 0.016] & [0.010, 0.011] & [0.002, 0.003] & [0.019, 0.168] & [0.013, 0.019] & [0.000, 0.010] &       &       &  \\
          \addlinespace
    \multicolumn{1}{l}{CR1(CF)} & 0.012 & 0.003 & 0.010 & 0.112 & 0.010 & 0.014 &       &       &  \\
          & [0.000, 0.020] & [0.002, 0.004] & [0.008, 0.011] & [0.005, 0.154] & [0.002, 0.016 & [0.009, 0.018] &       &       &  \\
          \addlinespace
    \multicolumn{1}{l}{CR2(CF)} & 0.012 & 0.004 & 0.010 & 0.113 & 0.011 & 0.014 &       &       &  \\
          & [-0.006, 0.015] & [0.002, 0.005] & [0.008, 0.011] & [-0.019, 0.172] & [0.002, 0.018] & [0.010, 0.021] &       &       &  \\
\cmidrule{1-7} 
\\[-0.5em]
\multicolumn{7}{l}{\footnotesize \textit{Notes}: CR: the convex regression with monotonicity;}\\
\multicolumn{7}{l}{\footnotesize \hspace*{3em} CR(IV): the 2SLS-based convex regression with monotonicity;}\\
\multicolumn{7}{l}{\footnotesize \hspace*{3em} CR1(CF): the IV control function-based convex regression with monotonicity;}\\
\multicolumn{7}{l}{\footnotesize \hspace*{3em} CR2(CF): the IV control function-based convex regression without monotonicity.}\\
\end{tabular}%
\label{tab:tab4}%
\end{table}%
\end{sidewaystable}

Compared to convex regression, the 2SLS-based convex regression leads to higher capital coefficients in all industries. For instance, the capital coefficients increase from a median of 0.004 to 0.012 for the Food Product industry and 0.002 to 0.01 for the Wood Products industry, respectively. This suggests that instrumenting capital input with lagged investment may yield a higher capital estimate, conforming with the findings in \citet{Collard2016}, even though a parametric estimation is employed.

After controlling for the impacts of simultaneity bias and mismeasurement in input, we again observe the higher median estimates for capital in all industries, regardless of whether monotonicity is considered. Likewise, for instance, the capital coefficients increase from a median of 0.004 to 0.006 (the IV control function approach with monotonicity) for the Food Product industry and 0.002 to 0.003 for the Wood Products industry, respectively. Furthermore, the capital median coefficients estimated by convex regression without monotonicity are larger than those by convex regression with monotonicity. The higher capital estimates may suggest that the IV control function is more effective in the absence of monotonicity, providing support for Theorems~\ref{theo2} and \ref{theo3}.

Regarding the labor estimate, we find that the median value decreases in all industries but Apparel, which seems to offset the increase in capital after using the IV control function approach or even the 2SLS estimation. For example, the median labor coefficient decreases from 0.029 by convex regression with monotonicity to 0.027 by the IV control function-based convex regression without monotonicity in the Food Products industry.

%-----------------%
%
%-----------------%

\section{Conclusions}\label{sec:conc}

This paper reviews and explains the sample orthogonality condition for the most standard specifications of convex regression. The orthogonality condition is critically important for understanding possible sources of endogeneity bias. The approach applied in this paper---leveraging the Lagrangian duality theory and orthogonality analysis---offers a potential framework for testing endogeneity of the nonparametric shape-constrained regression estimators.

Our analytical results indicate that the validity of the sample orthogonality condition in convex regression is not determined by the presence of a homogeneity constraint. The validity of the sample orthogonality condition depends solely on the first-order condition of the slope coefficients of the explanatory variables (i.e., $\bbeta_i$) and not on the intercept parameter (i.e., $\alpha_i$). Thus, a sign constraint on the intercepts of the regression hyperplanes does not affect the sample orthogonality condition of convex regression.

Imposing a sign constraint on the slope coefficients alters the sample orthogonality condition. Thus, a convex regression problem that includes a monotonicity condition violates the orthogonality condition. The covariance between the regression residuals and the explanatory variables is negative. Intuitively, this means that for higher values of the independent variables, more observation points fall below the estimated regression function.

The results for the convex regression with a multiplicative error term indicate that the usual sample orthogonality condition does not hold. However, when the explanatory variables are normalized by the fitted value (i.e., $\bx/f(\bx)$), then the sample orthogonality condition holds. Analogous to the additive model, the special case of the sample orthogonality condition is valid for convex regression with or without a linear homogeneity constraint. By imposing a monotonicity constraint, the sample orthogonality condition does not hold in the multiplicative model either.

Since the convex regression is not immune to endogeneity shown by the sample orthogonality condition analysis, we propose a hybrid IV control function approach to alleviate the endogeneity bias. The simulation results confirm the superiority of the proposed approach in endogeneity correction in comparison with the traditional 2SLS approach. The monotonicity is helpful in improving the accuracy of regression function estimation, even though it could introduce bias. The empirical application also suggests that the hybrid IV control function approach makes a difference in convex regression to mitigate the impact of endogeneity and that imposing monotonicity constraints can be useful in practice. 

%-----------------%
%
%-----------------%

\section*{Acknowledgments}\label{sec:ack}

The authors acknowledge CSC – IT Center for Science, Finland, for computational resources.

%-----------------%
%
%-----------------%

\baselineskip 12pt
\bibliographystyle{econ-econometrica}
\bibliography{References.bib} 

%-----------------%
%
%-----------------%

\clearpage
\newpage
\baselineskip 20pt
\section*{Appendix}\label{sec:app}

\renewcommand{\thesubsection}{\Alph{subsection}}
\renewcommand{\thetable}{B\arabic{table}}
\setcounter{table}{0}
\renewcommand{\thefigure}{B\arabic{figure}}
\setcounter{figure}{0}
\renewcommand{\theequation} {A.\arabic{equation}}
\setcounter{equation}{0}

%-----------------%
%
%-----------------%

\subsection{Proofs}\label{sec:proofs}

\textbf{Proof of Lemma \ref{prop1}}. By combining conditions \eqref{eq:eq5} and \eqref{eq:eq6}, there is $e_i = \dfrac{1}{2}\sum\limits_{h=1}^n(\lambda_{hi}-\lambda_{ih})$. By summing over all $n$ observations, we can confirm that $\sum\limits_{h=1}^n e_i=0$. \qed
\newline

\noindent\textbf{Proof of Theorem \ref{theo1}}. By combining \eqref{eq:eq5} and \eqref{eq:eq7} and summing over all $n$ observations, there is 
	\begin{equation}
		\sum\limits_{i=1}^ne_{i}x_{ik}=\sum\limits_{i=1}^n{\sum\limits_{h=1}^n{\lambda _{ih}x_{ik}}}-\sum\limits_{i=1}^n{\sum\limits_{h=1}^n{\lambda_{hi}x_{hk}}}
		\label{eq:a1}
	\end{equation} 
	By utilizing the properties of the sum operator to substitute indices $i$ and $h$ we know that $\sum\limits_{i=1}^n{\sum\limits_{h=1}^n{\lambda _{ih}x_{ik}}}=\sum\limits_{i=1}^n{\sum\limits_{h=1}^n{\lambda _{hi}x_{hk}}}$ and hence $\sum\limits_{i=1}^ne_ix_{ik}=0$.\qed
\newline

\noindent\textbf{Proof of Lemma \ref{prop2}}. Combining conditions \eqref{eq:eq5} and \eqref{eq:eq16}, we have $e_i=\hat{y}_{i}\left[ \sum\limits_{h=1}^n\lambda_{ih}-\sum\limits_{h=1}^n\lambda_{hi} \right]$. We then move $\hat{y}_{i}$ to the left-hand side and sum over all $n$ observations. Since $\sum\limits_{i=1}^n \sum\limits_{h=1}^n\lambda_{ih}-\sum\limits_{i=1}^n\sum\limits_{h=1}^n\lambda_{hi} = 0$, we can obtain $\sum\limits_{i=1}^n \dfrac{e_i}{\hat{y}_{i}} = 0$. \qed
\newline

\noindent\textbf{Proof of Theorem \ref{theo2}}. 
Summing over all $n$ observations in both~\eqref{eq:eq5} and \eqref{eq:eq17}, we have 
\begin{align}
	& \sum\limits_{i=1}^n e_i - \sum\limits_{i=1}^n\mu_i=0 \label{eq:a2}\\
	& -\sum\limits_{i=1}^n\mu_i\dfrac{x_{ik}}{\hat{y}_i}+\sum\limits_{i=1}^n\sum\limits_{h=1}^n \lambda_{ih} x_{ik} - \sum\limits_{i=1}^n\sum\limits_{h=1}^n \lambda_{hi}x_{hk} = 0 \label{eq:a3}
\end{align}
After combining Eqs.~\eqref{eq:a2} and \eqref{eq:a3}, the sample orthogonality condition becomes $\sum\limits_{i=1}^n e_i\dfrac{x_{ik}}{\hat{y}_i} =0$. \qed
\newline

\noindent\textbf{Proof of Theorem \ref{theo3}}. Summing Eq.~\eqref{eq:eq13} (or Eq.~\eqref{eq:eq22}) over all observations, and applying Theorem~\ref{theo1} (or Theorem~\ref{theo2}), we have $\sum\limits_{i=1}^n \eta_{ik} > 0$ and hence $\sum\limits_{i=1}^n e_i x_{ik} < 0$ (or $\sum\limits_{i=1}^n e_i \dfrac{x_{ik}}{\hat{y}_i} < 0$). \qed
\newline

\noindent\textbf{Proof of Theorem \ref{theo4}}. After considering the additive model with linear homogeneity constraints, we can eliminate the $\alpha_i$ variables from the problem as $\alpha_i = 0 (i=i, \ldots, n)$. The Lagrangian function is 
\begin{align*}
    \Ls(\blambda,\bbeta,\bmu)=&\dfrac{1}{2}\sum\limits_{i=1}^n \varepsilon_i^2+\sum\limits_{i=1}^n\mu_i(y_i-\bbeta_i^\prime\bx_i-\varepsilon_i) + \sum\limits_{i=1}^n{\sum\limits_{h=1}^n\lambda_{ih}(\bbeta _i-\bbeta_h)^\prime\bx_i}
\end{align*}

The corresponding first-order conditions are the following equations 
\begin{align}
& \dfrac{\partial \Ls}{\partial \varepsilon_i}=e_i - \mu_i=0, \forall I \label{eq:a4} \\
& \dfrac{\partial \Ls}{\partial \beta_{ik}}=-\mu_i x_{ik}+\sum\limits_{h=1}^n \lambda_{ih} x_{ik} - \sum\limits_{h=1}^n\lambda_{hi} x_{hk}=0, \forall i, k.  \label{eq:a5}
\end{align}

The result demonstrated in Theorem \ref{theo2} is still valid as it is dependent on the equations \eqref{eq:a4} and \eqref{eq:a5}, which are still active. Therefore, the additive model subject to concavity and linear homogeneity holds the orthogonality condition. However, after introducing monotonicity to this convex regression problem, the orthogonality condition will not be held (see Theorem \ref{theo3} for the proof). Similarly, this conclusion can be directly applied to the multiplicative models. \qed
\newline

\noindent\textbf{Proof of Theorem \ref{theo5}}. The regressor $\hat{\zeta}_i$ in the semi-nonparametric convex regression model \eqref{eq:eq4.2} is derived from the residuals of an auxiliary OLS regression defined in \eqref{eq:eq4.1}. Under standard assumptions for OLS estimation---specifically, independence and identical distribution of observations, full rank condition of the regressors, zero conditional expectation, and finite variance of the random error term---the estimator $\hat{\omega}$ is known to be consistent:
\begin{equation}
    \hat{\omega} \overset{p}{\rightarrow} \omega.
\end{equation}

It then follows directly from Slutsky's theorem that the residual estimates $\hat{\zeta}_i = x_{1i}-\hat{\omega}z_i$ consistently estimate $\zeta_i$, specifically:
\begin{equation}
    \hat{\zeta}_i \overset{p}{\rightarrow} \zeta_i.
\end{equation}

Consequently, the consistency proof of the semi-nonparametric convex regression estimator \eqref{eq:eq4.3} can be rigorously established by adapting and extending the empirical process-based uniform convergence results, identification conditions, and consistency presented in \citet{Seijo2011} or \citet{Lim2012}. This extension explicitly accounts for the consistent estimation of the auxiliary residual term $\hat{\zeta}_i$, thus maintaining the theoretical rigor and validity of the consistency of convex regression.

\newpage
\subsection{Additional figures and tables}\label{sec:addfigures}

\begin{table}[H]
  \centering
  \caption{Bias comparison between OLS and convex regression with $n=400$ standard
deviations in parentheses).}
  \footnotesize
  \setlength{\tabcolsep}{4.5pt}
    \begin{tabular}{llccccccccc}
    \toprule
    \multicolumn{1}{l}{\multirow{3}[6]{*}{\begin{tabular}{l} Endogeneity \\ $\rho$ \end{tabular}}}
 & \multicolumn{1}{l}{\multirow{3}[6]{*}{\begin{tabular}{l} Noise \\ $\sigma_\varepsilon$ \end{tabular}}} & \multicolumn{4}{c}{DGP I: $k=2$} &       & \multicolumn{4}{c}{DGP II: $k=3$} \\
\cmidrule{3-6}\cmidrule{8-11}          &       & \multicolumn{2}{c}{OLS} & \multicolumn{2}{c}{CR} &       & \multicolumn{2}{c}{OLS} & \multicolumn{2}{c}{CR} \\
\cmidrule{3-6}\cmidrule{8-11}          &       & \multicolumn{1}{c}{logs} & \multicolumn{1}{c}{levels} & \multicolumn{1}{c}{conc+mon} & \multicolumn{1}{c}{conc} &       & \multicolumn{1}{c}{logs} & \multicolumn{1}{c}{levels} & \multicolumn{1}{c}{conc+mon} & \multicolumn{1}{c}{conc} \\
    \midrule
    0     & 0.5   & 0.00  & 3.39  & -0.06 & -0.03 &       & 0.01  & 3.33  & 0.08  & 0.20 \\
          &       & (0.68)  & (0.83)  & (0.67)  & (0.67)  &       & (0.64)  & (0.78)  & (0.64)  & (0.64) \\
          & 1     & 0.03  & 16.53 & 0.03  & 0.12  &       & 0.08  & 16.26 & 0.33  & 0.77 \\
          &       & (1.36)  & (3.08)  & (1.54)  & (1.33)  &       & (1.28)  & (2.80)  & (1.28)  & (1.30) \\
          & 2     & 0.25  & 161.48 & 0.29  & 0.67  &       & 0.33  & 147.31 & 0.86  & 2.14 \\
          &       & (2.77)  & (62.70) & (2.69)  & (2.70)  &       & (2.51)  & (71.35) & (2.53)  & (2.71) \\
    0.9   & 0.5   & 5.30  & 7.00  & 3.18  & 3.19  &       & 4.09  & 5.51  & 2.92  & 2.96 \\
          &       & (0.90)  & (1.16)  & (0.61)  & (0.61)  &       & (0.39)  & (0.44)  & (0.35)  & (0.35) \\
          & 1     & 14.74 & 23.86 & 5.32  & 5.41  &       & 12.55 & 20.68 & 5.91  & 6.28 \\
          &       & (3.50)  & (5.61)  & (1.43)  & (1.44)  &       & (1.11)  & (1.64)  & (0.79)  & (0.79) \\
          & 2     & 56.64 & 152.92 & 7.86  & 8.34  &       & 40.88 & 111.04 & 7.68  & 9.28 \\
          &       & (11.01) & (32.33) & (2.04)  & (2.11)  &       & (19.39) & (53.59) & (3.85)  & (4.58) \\
    \bottomrule
    \\[-0.8em]
    \multicolumn{11}{l}{\footnotesize \textit{Notes}: logs: OLS regressions in logarithmic form; levels: OLS regressions in level form;} \\
    \multicolumn{11}{l}{\footnotesize \hspace*{1cm} conc+mono: the multiplicative model subject to monotonicity and concavity;}\\
    \multicolumn{11}{l}{\footnotesize \hspace*{1cm} conc: the multiplicative model subject to concavity.}\\
    \end{tabular}%
  \label{tab:a1}%
\end{table}%

\begin{figure}[H]
    \centering
    \includegraphics[width=0.9\linewidth]{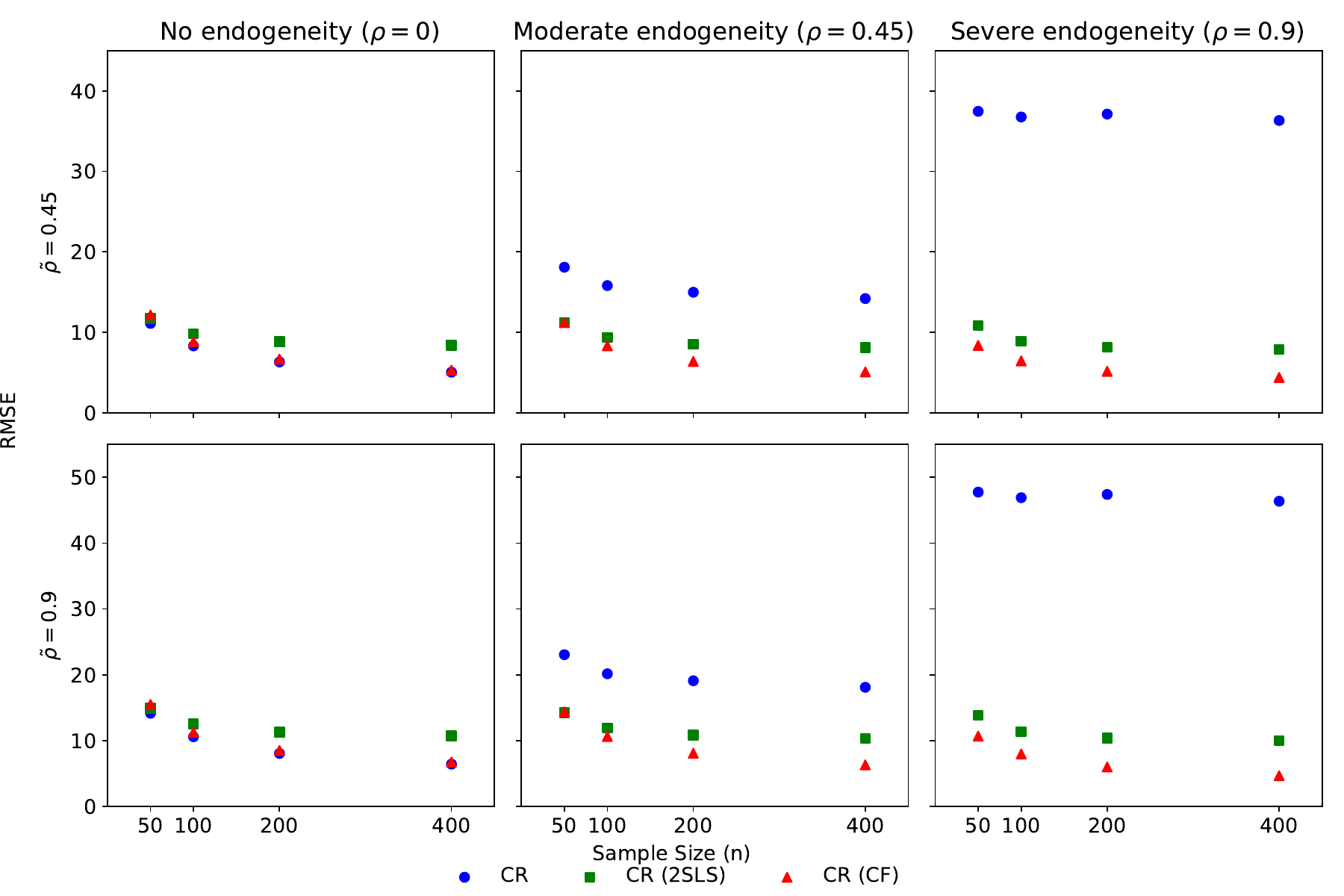}
    \vspace{-.2cm}
    \caption{The RMSE statistic of convex regression estimators with $\sigma=1$ and $k=3$.}
    \label{fig:figa1}
\end{figure}
\vspace{-0.5cm}

\begin{figure}[H]
    \centering
    \includegraphics[width=0.9\linewidth]{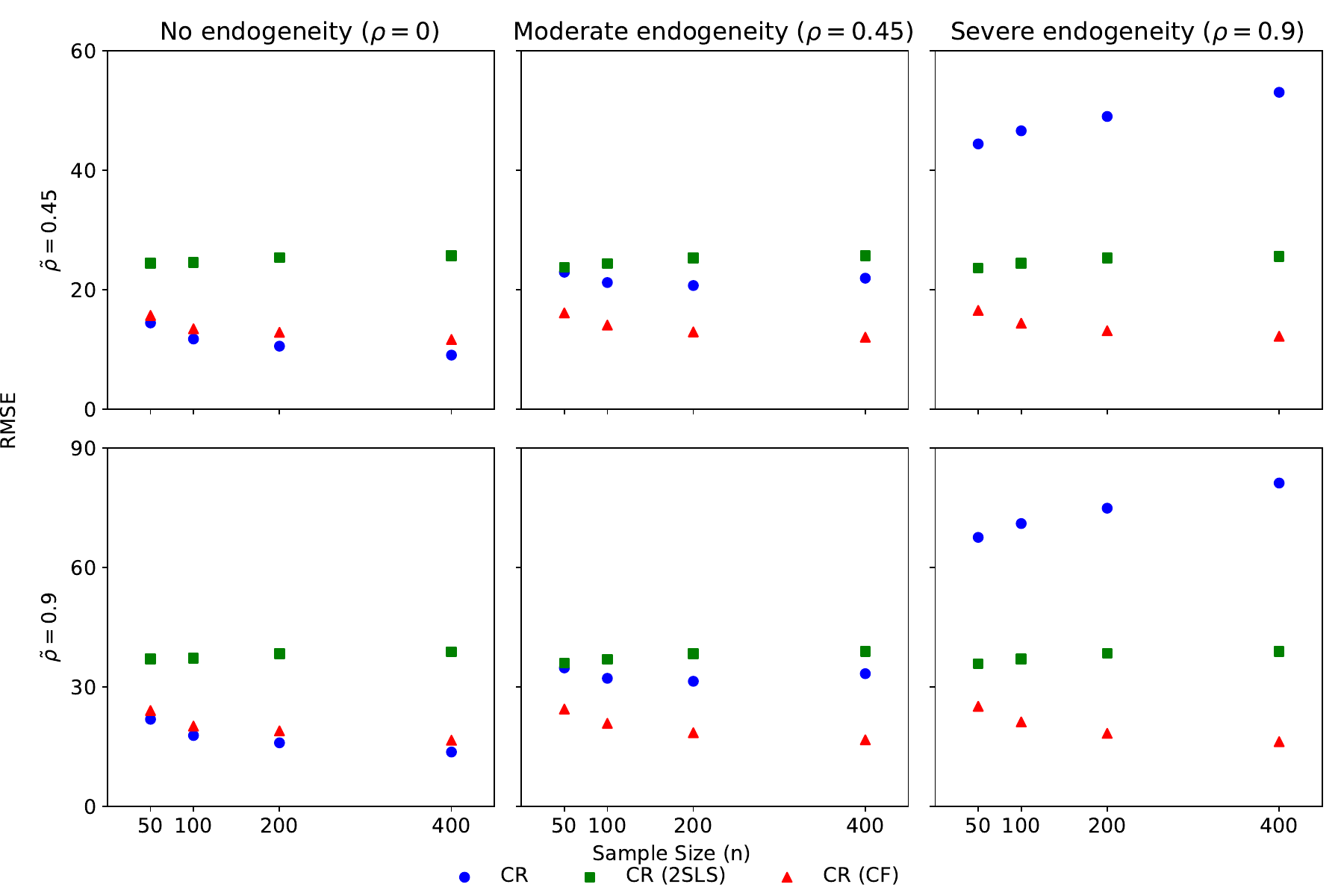}
    \vspace{-.2cm}
    \caption{The RMSE statistic of convex regression estimators in the case of monotonicity with $\sigma_\varepsilon=2$ and $k=2$.}
    \label{fig:figa2}
\end{figure}

\end{document}